\documentclass[onefignum,onetabnum,onethmnum]{siamart171218}

\usepackage{amsfonts,amsmath,amssymb}
\usepackage{mathtools}
\usepackage{thmtools,thm-restate}
\usepackage[noend]{algorithmic}
\algsetup{indent=2em}

\usepackage{tikz}
\usetikzlibrary{arrows,calc,automata,shapes}
\tikzset{AUT style/.style={thick,>=angle 60,every edge/.append style={thick},every state/.style={thick,minimum size=20,inner sep=0.5},states/.style={state,ellipse,minimum height=40,minimum width=30}}}
\tikzset{DFA style/.style={>=angle 60,every edge/.append style={thick},every state/.style={thick,rounded rectangle, minimum size=20,inner sep=0.5}}}

\declaretheorem[name=Theorem,Refname={Theorem,Theorems},sibling=theorem]{ourtheorem}
\declaretheorem[name=Lemma,Refname={Lemma,Lemmas},sibling=theorem]{ourlemma}

\declaretheorem[name=Example,Refname={Example,Examples},sibling=theorem]{ourexample}

\newcommand{\A}{\mathcal{A}}
\newcommand{\M}{\mathcal{M}}
\newcommand{\MM}{\overline{\mathcal{M}}}
\newcommand{\N}{\mathbb{N}}
\newcommand{\ourbound}{\frac1{16} n^5 + \frac{15}{16} n^4}
\newcommand{\Q}{\mathbb{Q}}
\newcommand{\R}{\mathbb{R}}
\newcommand{\rk}{\mathit{rk}}
\newcommand{\winit}{w_{\mathit{init}}}
\newcommand{\Z}{\mathbb{Z}}

\newcommand{\norm}[1]{\left\lVert#1\right\rVert}

\headers{On Nonnegative Integer Matrices and Short Killing Words}{Stefan Kiefer and Corto Mascle}
\title{On Nonnegative Integer Matrices and Short Killing Words\thanks{A preliminary version of this article is appearing at STACS'19 under the title \emph{On Finite Monoids over Nonnegative Integer Matrices and Short Killing Words}.
This article is more self-contained and slightly generalizes the results by relaxing the finiteness condition to a condition on the joint spectral radius.
In addition we prove a more precise result related to Restivo's conjecture for finite codes.
\funding{Kiefer is supported by a Royal Society University Research Fellowship.}}}

\author{Stefan Kiefer\thanks{University of Oxford, UK
  (\url{https://www.cs.ox.ac.uk/people/stefan.kiefer/}).}
\and Corto Mascle\thanks{ENS Paris-Saclay, France.}}

\ifpdf
\hypersetup{
  pdftitle={On Nonnegative Integer Matrices and Short Killing Words},
  pdfauthor={S. Kiefer and C. Mascle}
}
\fi

\begin{document}

\sloppy
\maketitle

\begin{abstract}
Let $n$ be a natural number and $\M$ a set of $n \times n$-matrices over the nonnegative integers such that the joint spectral radius of $\M$ is at most one.
We show that if the zero matrix~$0$ is a product of matrices in~$\M$, then there are $M_1, \ldots, M_{n^5} \in \M$ with $M_1 \cdots M_{n^5} = 0$.
This result has applications in automata theory and the theory of codes.
Specifically, if $X \subset \Sigma^*$ is a finite incomplete code, then there exists a word $w \in \Sigma^*$ of length polynomial in $\sum_{x \in X} |x|$ such that $w$ is not a factor of any word in $X^*$.
This proves a weak version of Restivo's conjecture.
\end{abstract}



\section{Introduction} \label{sec-intro}

Let $n \in \N$ and $\M \subseteq \R^{n \times n}$ be a finite set of matrices.
The \emph{joint spectral radius} of~$\M$, denoted by~$\rho(\M)$, is defined by the following limit:
\[
 \rho(\M) \ := \ \lim_{k \to \infty} \max \{\norm{M_1 \cdots M_k}^{1/k} : M_i \in \M\}
\]
This limit exists and does not depend on the chosen norm~\cite{DaubechiesLagarias01}.
In this article we focus on nonnegative integer matrices: we assume $\M \subseteq \N^{n \times n}$ where $\N = \{0, 1, 2, \ldots\}$.
Denote by~$\MM$ the monoid (semigroup) generated by~$\M$ under matrix multiplication, i.e., the set of products of matrices from~$\M$.
If $\MM$ is finite then $\rho(\M) \le 1$, but the converse does not hold~\cite{Jungers08}.

In this article we show the following theorem:
\begin{restatable}{ourtheorem}{thmmain} \label{thm-main}
Let $n \in \N$ and $\M \subseteq \N^{n \times n}$ be a finite set of nonnegative integer matrices with $\rho(\M) \le 1$.
Then there are $M_1, \ldots, M_\ell \in \M$ with $\ell \le \ourbound$ such that the matrix product $M_1 \cdots M_\ell$ has minimum rank among the matrices in~$\MM$. 
Further, $M_1, \ldots, M_\ell$ can be computed in time polynomial in the description size of~$\M$.
\end{restatable}

\begin{ourexample} \label{ex-intro}
Let $n = 3$ and $\M = \{A, B\}$ where
\[
A \ = \
\begin{pmatrix}
0 && 0 && 1 \\
1 && 0 && 1 \\
0 && 0 && 0
\end{pmatrix} \quad \text{and} \quad
B \ = \
\begin{pmatrix}
0 && 0 && 0 \\
0 && 1 && 0 \\
0 && 1 && 0
\end{pmatrix}\,.
\]
Then $\MM$ is finite and $\rho(\M) = 1$.
Further, the matrix product $A A A$ is the zero matrix and, hence, has rank~$0$.
No other product of length~$3$ yields the zero matrix.
\end{ourexample}

Let $\M \subseteq \N^{n \times n}$ be a finite set of nonnegative integer matrices.
For notational convenience, throughout the paper, we associate to~$\M$ a bijection $M : \Sigma \to \M$ and extend it to the monoid morphism $M : \Sigma^* \to \MM$, where $\Sigma^* = \bigcup_{i=0}^\infty \Sigma^i$ denotes the set of \emph{words} over~$\Sigma$.
For a word $w \in \Sigma^i$, its \emph{length} $|w|$ is~$i$.
We write $\varepsilon$ for the word of length~$0$.
We may write $M(\Sigma)$ for~$\M$ and $M(\Sigma^*)$ for~$\MM$ and $\rho(M)$ for $\rho(M(\Sigma))$.
Then one may rephrase the main theorem as follows:
\newcommand{\stmtthmmainrephrased}{Given $M : \Sigma \to \N^{n \times n}$ with $\rho(M) \le 1$, one can compute in polynomial time a word $w \in \Sigma^*$ with $|w| \le \ourbound$ such that $M(w)$ has minimum rank in~$M(\Sigma^*)$.
}
\renewcommand{\thmcontinues}[1]{rephrased}
\begin{ourtheorem}[continues=thm-main]
\stmtthmmainrephrased
\end{ourtheorem}
The condition $\rho(M) \le 1$ should be viewed in light of the following dichotomy~\cite{Jungers08}: if $\rho(M) \le 1$ then $B(k) := \max \{\norm{M(w)} : w \in \Sigma^k\}$ is in $O(k^n)$, i.e., $B(k)$ grows polynomially in~$k$; if $\rho(M) > 1$ then (by definition) $B(k)$ grows exponentially in~$k$.

\subsubsection*{Automata definitions}

A 
function $M : \Sigma \to \{0,1\}^{n \times n}$
is naturally associated with an \emph{automaton}.
A \emph{nondeterministic finite automaton (NFA)} is a triple $\A = (\Sigma, Q, \delta)$, where $\Sigma$ is a finite alphabet, $Q$ is a finite set of states, and $\delta: Q \times \Sigma \to 2^{Q}$ is a transition function (initial and final states do not play a role here).
We extend $\delta$ in the usual way to $\delta : 2^Q \times \Sigma^* \to 2^Q$ by setting $\delta(P,a) := \bigcup_{q \in P} \delta(q,a)$ and $\delta(P,\varepsilon) := P$ and $\delta(P,w a) := \delta(\delta(P,w),a)$, where $P \subseteq Q$ and $a \in \Sigma$ and $w \in \Sigma^*$.
A sequence $\psi = q_0 a_1 q_1 a_2 \cdots q_{n-1} a_n q_n$ with $q_i \in Q$ and $a_i \in \Sigma$ is called a \emph{path} from $q_0$ to~$q_n$ if $\delta(q_{i-1},a_i) \ni q_i$ holds for all $i \in \{1, \ldots, n\}$.
The word $a_1 \cdots a_n$ is said to \emph{label} the path~$\psi$.
Note that a word $w \in \Sigma^*$ labels a path from $p$ to~$q$ if and only if $\delta(\{p\},w) \ni q$.
A word~$w$ is called \emph{killing word} if it does not label any path.
Associate to~$\A$ the monoid morphism $M_\A : \Sigma^* \to \N^{Q \times Q}$ where for all $a \in \Sigma$ we define $M_\A(a)(p,q) = 1$ if $\delta(p,a) \ni q$ and $0$ otherwise.
Then, for any word $w \in \Sigma^*$ we have that $M_\A(w)(p,q)$ is the number of $w$-labelled paths from $p$ to~$q$.
In particular, $M_\A(w)$ is the zero matrix~$0$ if and only if $w$ is a killing word.

An NFA $\A = (\Sigma, Q, \delta)$ is called an \emph{unambiguous finite automaton (UFA)} if for all states $p,q$ all paths from $p$ to~$q$ are labelled by different words, i.e., for each word $w \in \Sigma^*$ there is at most one $w$-labelled path from $p$ to~$q$.
Call a monoid $\MM \subseteq \N^{n \times n}$ an \emph{unambiguous monoid of relations} if $\MM \subseteq \{0,1\}^{n \times n}$.
For every UFA~$\A$ the monoid $M_\A(\Sigma^*)$ is an unambiguous monoid of relations, and every unambiguous monoid of relations can be viewed as generated by a UFA.
UFAs play a central role in our proofs.

\subsubsection*{The mortality problem}

\Cref{thm-main} is related to the \emph{mortality} problem for integer matrices: given $M : \Sigma \to \Z^{n \times n}$, is $0 \in M(\Sigma^*)$, i.e., can the zero matrix (which is defined to have rank~$0$) be expressed as a finite product of matrices in~$M(\Sigma)$?
Paterson~\cite{Paterson70} showed that the mortality problem for integer matrices is undecidable for $n = 3$.
It remains undecidable for $n = 3$ with $|\Sigma| = 7$ and for $n = 21$ with $|\Sigma| = 2$, see~\cite{HalavaHH07}.
Mortality for $n = 2$ is NP-hard~\cite{Bell12Mortality} and not known to be decidable, see \cite{Potapov17} for recent work on $n=2$.

The mortality problem for \emph{nonnegative} 
matrices (even for matrices over the nonnegative reals) is much easier, as for each matrix entry it only matters whether it is zero or nonzero, so one can assume $M : \Sigma \to \{0,1\}^{n \times n}$.
It follows that the mortality problem for nonnegative matrices is equivalent to the problem whether an NFA has a killing word.
The problem is PSPACE-complete~\cite{Shallit09}, and there are examples where the shortest killing word has exponential length in the number of states of the automaton~\cite{game82,Shallit09}.
This implies that the assumption in \cref{thm-main} about the joint spectral radius~$\rho(M)$ cannot be dropped.
Whether $\rho(M) \le 1$ indeed holds can be checked in polynomial time~\cite{Jungers08}.
The condition is satisfied whenever $M(\Sigma^*)$ is finite.
Whether $M(\Sigma^*)$ is finite can also be checked in polynomial time, see, e.g., \cite{WeberSeidlITA} and the references therein.
The authors are not aware of an easier proof of \cref{thm-main} under the stronger assumption that $M(\Sigma^*)$ is finite.
If $\rho(M) \le 1$ then the mortality problem for nonnegative integer matrices is solvable in polynomial time:
\begin{restatable}{ourproposition}{propmortality} \label{prop-mortality}
Given $M : \Sigma \to \N^{n \times n}$ with $\rho(M) \le 1$, one can decide in polynomial time if $0 \in M(\Sigma^*)$.
\end{restatable}
\Cref{prop-mortality} is implied by \cref{thm-main}, but has an easier proof.

\subsubsection*{Short killing words for unambiguous finite automata}
\Cref{prop-mortality} provides a polynomial-time procedure for checking
whether a UFA has a killing word.
Define $\rho$ as the spectral radius of the rational matrix $\frac{1}{|\Sigma|} \sum_{a \in \Sigma} M(a)$.
One can show that $\rho < 1$ if $\A$ has a killing word, and $\rho = 1$ otherwise (\cref{lem-spectral-radius}).
\cref{prop-mortality} then follows from the fact that one can compare $\rho$ with~$1$ in polynomial time.
Thus the spectral radius tells whether there \emph{exists} a killing word, but does not \emph{provide} a killing word.
Neither does this method imply a polynomial bound on the length of a minimal killing word, let alone a polynomial-time algorithm for computing a killing word.
\cref{thm-main}, which is proved purely combinatorially, fills this gap: if there is a killing word, then one can compute a killing word of length $O(|Q|^5)$ in polynomial time.
NP-hardness results for approximating the length of a shortest killing word were proved in~\cite{RyzhikovSzykula18}, even for the case $|\Sigma|=2$ and for \emph{partial DFAs}, which are UFAs with $|\delta(p,a)| \le 1$ for all $p \in Q$ and all $a \in \Sigma$.
In fact, by combining our main result with~\cite[Theorem~17]{RyzhikovSzykula18} the following problem is NP-complete: given an unambiguous automaton and a number $\ell \in \N$ in binary, does there exist a killing word of length at most~$\ell$?

\subsubsection*{Short minimum-rank words}

Define the \emph{rank} of a UFA $\A = (\Sigma, Q, \delta)$ as the minimum rank of the matrices $M_\A(w)$ for $w \in \Sigma^*$.
A word~$w$ such that the rank of $M_\A(w)$ attains that minimum is called a \emph{minimum-rank} word.
Minimum-rank words have been very well studied for deterministic finite automata (DFAs).
DFAs are UFAs with $|\delta(p,a)| = 1$ for all $p \in Q$ and all $a \in \Sigma$.
In DFAs of rank~$1$, minimum-rank words are called \emph{synchronizing} because $\delta(Q,w)$ is a singleton when $w$ is a minimum-rank word.
It is the famous {\v{C}}ern{\'y} conjecture that whenever a DFA has a synchronizing word then it has a synchronizing word of length at most $(n-1)^2$ where $n := |Q|$.
There are DFAs whose shortest synchronizing words have that length, but the best known upper bound is cubic in~$n$, see~\cite{Volkov08} for a survey on the {\v{C}}ern{\'y} conjecture.

In 1986 Berstel and Perrin generalized the {\v{C}}ern{\'y} conjecture from DFAs to UFAs by conjecturing~\cite{BerstelPerrin86} that in any UFA a shortest minimum-rank word has length $O(n^2)$.
They remarked that no polynomial upper bound was known.
Then Carpi~\cite{Carpi88} showed the following:%
\begin{ourtheorem}[Carpi~\cite{Carpi88}] \label{thm-Carpi}
Let $\A = (\Sigma, Q, \delta)$ be a UFA of rank $r \ge 1$ such that the state transition graph of~$\A$ is strongly connected.
Let $n := |Q| \ge 1$.
Then $\A$ has a minimum-rank word of length at most $\frac12 r n (n-1)^2 + (2 r - 1) (n-1)$.
\end{ourtheorem}
This implies an $O(n^4)$ bound for the case where $r \ge 1$.
Carpi left open the case $r=0$, i.e., when a killing word exists.
The main technical contribution of our paper concerns the case $r=0$.
Combined with Carpi's \cref{thm-Carpi} we then obtain \cref{thm-main}.
Based on our technical development, we also provide a short proof of a variant of Carpi's \cref{thm-Carpi}, which suffices for our purposes and makes this article self-contained.
\cref{thm-main} provides, to the best of the authors' knowledge, the first polynomial bound, $O(n^5)$, on the length of  shortest minimum-rank words for UFAs.

\subsubsection*{Restivo's conjecture}
Let $X \subseteq \Sigma^*$ be a finite set of words over a finite alphabet~$\Sigma$, and define $k := \max_{x \in X} |x|$.
A word $v \in \Sigma^*$ is called \emph{uncompletable} in~$X$ if there are no words $u, w \in \Sigma^*$ such that $u v w \in X^*$, i.e., $v$ is not a factor of any word in $X^*$.
In~1981 Restivo~\cite{Restivo81} conjectured that if there exists an uncompletable word then there is an uncompletable word of length at most $2 k^2$.
This strong form of Restivo's conjecture was refuted in~\cite{Gusev11}, with a lower bound of $5 k^2 - O(k)$.
See also~\cite{Carpi17} for more recent work and open problems related to Restivo's conjecture.
The article~\cite{Julia17} describes a sophisticated computer-assisted search for sets~$X$ with long shortest uncompletable words.
While these experiments did not formally disprove a quadratic upper bound in~$k$, they seemed to hint at an exponential behaviour in~$k$.
Indeed, in a recent preprint~\cite{MikaSzykula19-arxiv} a lower bound of $2^{k/4}\cdot k/4$ was given, refuting Restivo's conjecture fundamentally.
The article~\cite{MikaSzykula19-arxiv} also provides a lower bound of $2^{\Omega(m^{1/5})}$, where $m := \sum_{x \in X} |x|$.

A set $X \subseteq \Sigma^*$ is called a \emph{code} if every word $w \in X^*$ has at most one decomposition $w = x_1 \cdots x_\ell$ with $x_1, \ldots, x_\ell \in X$.
See~\cite{BerstelCodesAutomata} for a comprehensive reference on codes.
For a finite code $X \subseteq \Sigma^*$ define $m := \sum_{x \in X} |x|$.
Given such~$X$ one can construct a \emph{flower automaton} \cite[Chapter 4.2]{BerstelCodesAutomata}, which is a UFA $\A_X = (\Sigma, Q, \delta)$ with $m - |X| + 1$ states, see \cref{fig-flower}.
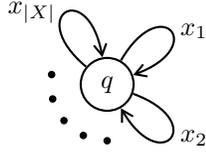
\begin{figure}
\begin{center}
\begin{tikzpicture}[xscale=3,yscale=3,AUT style]
\node[state] (q) at (0,0) {$q$};
\path[->] (q) edge [loop,out=140,in=100,looseness=11] node[left] {$x_{|X|}$} (q);
\path[->] (q) edge [loop,out=60,in=20,looseness=11] node[right] {$x_1$} (q);
\path[->] (q) edge [loop,out=-20,in=-60,looseness=11] node[right] {$x_2$} (q);
\foreach \i in {-25,0,25,-50,50}{
\node[circle,fill,inner sep=1.0] at (220+\i:0.25) {};
}
\end{tikzpicture}
\end{center}
\caption{Given a finite language $X \subseteq \Sigma^*$, the flower automaton $\A_X$ has one ``petal'' for each word $x \in X$. Thus $\delta(q, w) \ni q$ holds if and only if $w \in X^*$.
If $X$ is a code then $\A_X$ is unambiguous.}
\label{fig-flower}
\end{figure}
In this UFA any word is killing if and only if it is uncompletable in~$X$.
Hence \cref{thm-main} implies an $O(m^5)$  bound on the length of the shortest uncompletable word in a finite code.
This proves a weak (note that $m^5$~may be much larger than $k^2$) version of Restivo's conjecture for finite codes.
By adapting our main argument so that it exploits the special structure of flower automata, we get a better result:
\begin{restatable}{ourtheorem}{thmrestivo} \label{thm-restivo}
Let $X \subseteq \Sigma^*$ be a finite code that has an uncompletable word.
Define $k := \max_{x \in X} |x|$ and $m := \sum_{x \in X} |x|$ and assume $k > 0$.
Then one can compute in polynomial time an uncompletable word of length at most $(k+1) k^2 (m+2) (m+1)$.
\end{restatable}
This result does not contradict the work~\cite{MikaSzykula19-arxiv}, as their sets~$X$ are not codes.
Contrasting the results of~\cite{MikaSzykula19-arxiv} with our \cref{thm-restivo}, we highlight as open problem the following version of Restivo's conjecture for finite codes: does every finite code with an uncompletable word have an uncompletable word of length polynomial in~$k$?

\subsubsection*{Is any product a short product?\nopunct}
It was shown in~\cite{WeberSeidlITA} that if $M(\Sigma^*) \subseteq \N^{n \times n}$ is finite then for every $w_0 \in \Sigma^*$ there exists $w \in \Sigma^*$ with $|w| \le \left\lceil e^2 n!\right\rceil-2$ such that $M(w_0) = M(w)$.
It was also shown in~\cite{WeberSeidlITA} that such a length bound cannot be smaller than $2^{n-2}$.
In view of \cref{thm-main} one may ask if a polynomial length bound exists for \emph{low-rank} matrices~$M(w_0)$.
The answer is no, even for unambiguous monoids of relations and even when $M(w_0)$ has rank~$1$ and $1$ is the minimum rank in~$M(\Sigma^*)$:
\begin{restatable}{ourtheorem}{thmlongproduct} \label{thm-long-product}
There is no polynomial~$p$ such that the following holds:
\begin{quote}
Let $M : \Sigma^* \to \{0,1\}^{n \times n}$ be a monoid morphism.
Let $w_0 \in \Sigma^*$ be such that $M(w_0)$ has rank~$1$, and let $1$ be the minimum rank in $M(\Sigma^*)$.
Then there is $w \in \Sigma^*$ with $|w| \le p(n)$ such that $M(w_0) = M(w)$.
\end{quote}
\end{restatable}
Thus, while \cref{thm-main} guarantees that \emph{some} minimum-rank matrix in the monoid is a short product, this is not the case for every minimum-rank matrix in the monoid.

\subsubsection*{By how much could the $O(n^5)$ upper bound be improved?\nopunct}
A \emph{synchronizing $0$-automaton} is a DFA $\A = (\Sigma, Q, \delta)$ that has a state $0 \in Q$ and a word $w \in \Sigma^*$ such that $\delta(Q,w x) = \{0\}$ holds for all $x \in \Sigma^*$.
The shortest such synchronizing words~$w$ are exactly the shortest killing words in the partial DFA obtained from~$\A$ by omitting all transitions into the state~$0$.
There exist synchronizing $0$-automata with $n$~states where the shortest synchronizing word has length $n(n-1)/2$, and $\frac{n^2}{4} + \Omega(n)$ lower bounds exist even for synchronizing $0$-automata with $|\Sigma| = 2$ \cite{Martugin08,JALC-2019-153}.
This implies that the $O(n^5)$ upper bound from \cref{thm-main} cannot be improved to $o(n^2)$, not even when a killing word exists.
One might generalize the {\v{C}}ern{\'y} conjecture by claiming \cref{thm-main} with an upper bound of $(n-1)^2$ (note that such a conjecture would concern minimum-rank words, not minimum nonzero-rank words).
To the best of the authors' knowledge, this vast generalization of the {\v{C}}ern{\'y} conjecture 
has not yet been refuted.

\subsubsection*{Organization of the article}
In the remaining four sections we prove \cref{prop-mortality,thm-main,thm-restivo,thm-long-product}, respectively.

\section{Proof of \texorpdfstring{\cref{prop-mortality}}{Proposition~\ref{prop-mortality}}} \label{sec-proof-mortality}

Let $M : \Sigma \to \N^{n \times n}$ be such that $\rho(M) \le 1$.

Towards a proof of \cref{prop-mortality}, define the rational nonnegative matrix $A \in \Q^{n \times n}$ by $A := \frac{1}{|\Sigma|} \sum_{a \in \Sigma} M(a)$.
Observe that for $k \in \N$ we have $A^k = \frac{1}{|\Sigma^k|} \sum_{w \in \Sigma^k} M(w)$, i.e., $A^k$ is the average of the $M(w)$, where $w$ ranges over all words of length~$k$.
Define $\rho \ge 0$ as the spectral radius of~$A$.
\begin{ourlemma} \label{lem-rho<=1}
We have $\rho \le 1$.
\end{ourlemma}
\begin{proof}
By the Perron-Frobenius theorem, $A$ has a nonnegative eigenvector $u \in \R^{n}$ with $A u = \rho u$.
So $A^k u = \rho^k u$.
Thus $\max \{\norm{M(w)} : w \in \Sigma^k\} \in \Omega(\rho^k)$.
Hence $\rho \le \rho(M) \le 1$.
\end{proof}
\begin{ourlemma} \label{lem-spectral-radius}
We have $\rho < 1$ if and only if there is $w \in \Sigma^*$ with $M(w) = 0$.
\end{ourlemma}
\begin{proof}
Suppose $\rho < 1$.
Then $\lim_{k \to \infty} A^k = 0$, and so there is $k \in \N$ such that the sum of all entries of $A^k$ is less than~$1$.
It follows that there is $w \in \Sigma^k$ such that the sum of all entries of~$M(w)$ is less than~$1$.
Since $M(w) \in \N^{n \times n}$ it follows $M(w) = 0$.

Conversely, suppose there is $w_0 \in \Sigma^*$ with $M(w_0) = 0$.
Since $\rho(M) \le 1$, by~\cite[Theorem~3]{Jungers08} there exists $c > 0$ such that $B(k) := \max \{\norm{M(w)} : w \in \Sigma^k\} \le c k^n$ holds for all $k \in \N \setminus \{0\}$.
For any $k \in \N$ define
$
 W(k) := \Sigma^k \setminus \left(\Sigma^* w_0 \Sigma^*\right)
$,
i.e., $W(k)$ is the set of length-$k$ words that do not contain $w_0$ as a factor.
Note that $M(w) = 0$ holds for all $w \in \Sigma^k \setminus W(k)$.
Since matrix norms are sub-additive, it follows that $\norm{A^k}$ is at most $\frac{|W(k)|}{|\Sigma^k|} \cdot B(k)$.
On the other hand, for any $m \in \N$, if a word of length $m |w_0|$ is picked uniformly at random, then the probability of picking a word in $W(m |w_0|)$ is at most
\[
 \left(1 - \frac{1}{|\Sigma^{|w_0|}|}\right)^m\,,
\]
thus
\[
 \norm{A^{m |w_0|}} \ \le \ \left(1 - \frac{1}{|\Sigma^{|w_0|}|}\right)^m c (m |w_0|)^n\,.
\]
Hence $\lim_{k \to \infty} A^k = 0$ and so $\rho < 1$.
\end{proof}

With these lemmas at hand, we can prove \cref{prop-mortality}:
\propmortality*
\begin{proof}
By \cref{lem-spectral-radius}, it suffices to check whether $\rho < 1$.

If $\rho < 1$ then the linear system $A x = x$ does not have a nonzero solution.
Conversely, if $\rho \ge 1$ then, by \cref{lem-rho<=1}, we have $\rho = 1$ and thus, by the Perron-Frobenius theorem, the linear system $A x = x$ has a (real) nonzero solution.

Hence it suffices to check if $A x = x$ has a nonzero solution.
This can be done in polynomial time.
\end{proof}

As remarked in \Autoref{sec-intro}, this algorithm does not exhibit a word~$w$ with $M(w) = 0$, even when it proves the existence of such~$w$.

\section{Proof of \texorpdfstring{\cref{thm-main}}{Theorem~\ref{thm-main}}} \label{sec-proof-main}

As before, let $M : \Sigma \to \N^{n \times n}$ be such that $\rho(M) \le 1$.
Call $M$ \emph{strongly connected} if for all $i,j \in \{1, \ldots, n\}$ there is $w \in \Sigma^*$ with $M(w)(i,j) \ge 1$.
In \cref{sub-sc} we consider the case that $M$ is strongly connected.
In \cref{sub-nsc} we consider the general case.

\subsection{Strongly Connected} \label{sub-sc}

In this section we consider the case that $M$ is strongly connected and prove the following proposition, which extends Carpi's \cref{thm-Carpi}:
\begin{restatable}{ourproposition}{propsc} \label{prop-sc}
Given $M : \Sigma \to \N^{n \times n}$ such that $\rho(M) \le 1$ and $M$ is strongly connected, one can compute in polynomial time a word $w \in \Sigma^*$ with $|w| \le \ourbound$ such that $M(w)$ has minimum rank in~$M(\Sigma^*)$.
\end{restatable}
In the strongly connected case, $M(\Sigma^*)$ does not have numbers larger than~$1$:

\begin{ourlemma} \label{lem-sc-means-UFA}
We have $M(\Sigma^*) \subseteq \{0,1\}^{n \times n}$.
\end{ourlemma}
\begin{proof}
Suppose $M(v)(i,j) \ge 2$ for some \mbox{$v \in \Sigma^*$}.
Since $M$ is strongly connected, there is $w \in \Sigma^*$ with $M(w)(j,i) \ge 1$.
Hence $M(v w)(i,i) \ge 2$.
It follows that $M((v w)^k)(i,i) \ge 2^k$ for all $k \in \N$, contradicting the assumption $\rho(M) \le 1$.
\end{proof}

\Cref{lem-sc-means-UFA} allows us to view the strongly connected case in terms of UFAs.
Define a UFA $\A = (\Sigma, Q, \delta)$ with $Q = \{1, \ldots, n\}$ and $\delta(p,a) \ni q$ if and only if $M(a)(p,q) = 1$.
For the rest of the subsection we will mostly consider~$Q$ as an arbitrary finite set of $n$ states. 
When there is no confusion, we may write $p w$ for $\delta(p,w)$ and $w q$ for $\{p \in Q : p w \ni q\}$.
We extend this to $P w := \bigcup_{p \in P} p w$ and $w P := \bigcup_{p \in P} w p$.
We say a state $p$ is \emph{reached by} a word~$w$ when $w p \ne \emptyset$, and a state $p$ \emph{survives} a word~$w$ when $p w \ne \emptyset$.
Note that $Q w$ is the set of states that are reached by~$w$, and $w Q$ is the set of states that survive~$w$.
Let $q_1 \ne q_2$ be two different states.
Then $q_1, q_2$ are called \emph{coreachable} when there is $w \in \Sigma^*$ with $w q_1 \cap w q_2 \ne \emptyset$ (i.e., there is $p \in Q$ with $p w \supseteq \{q_1, q_2\}$),
and they are called \emph{mergeable} when there is $w \in \Sigma^*$ with $q_1 w \cap q_2 w \ne \emptyset$.
For any $q \in Q$ we define $C(q)$ as the set of states coreachable with~$q$.
Also, define $c := \max \{|q w| : q \in Q,\ w \in \Sigma^*\}$ and $m := \max \{|w q| : w \in \Sigma^*, \ q \in Q\}$.
The following lemma says that one can compute short witnesses for coreachability:
\begin{ourlemma} \label{lem-product-automaton}
If states $q \ne q'$ are coreachable, then one can compute in polynomial time $w_{q,q'} \in \Sigma^*$ with $|w_{q,q'}| \le \frac12 (n+2) (n-1)$ such that $q w_{q,q'} \supseteq \{q, q'\}$.
\end{ourlemma}
\begin{proof}
Let $q \ne q'$ be coreachable states.
Then there are $p \in Q$ and $v \in \Sigma^*$ with $p v \supseteq \{q, q'\}$.
Since $M$ is strongly connected, there is $u \in \Sigma^*$ with $q u \ni p$, hence $q u v \supseteq \{q, q'\}$.
Define an edge-labelled directed graph $G = (V,E)$ with vertex set $V = \{\{r,s\} : r,s \in Q\}$ and edge set $E = \{(R,a,S) \in V \times \Sigma \times V: R a \supseteq S\}$.
Since $q u v \supseteq \{q, q'\}$, the graph~$G$ has a path, labelled by $u v$, from $\{q\}$ to $\{q,q'\}$.
The shortest path from $\{q\}$ to $\{q,q'\}$ has at most $|V|-1$ edges and is thus labelled with a word $w \in \Sigma^*$ with $|w| \le |V| - 1 = \frac12 n(n+1) - 1 = \frac12 (n+2) (n-1)$.
For this~$w$ we have $q w \supseteq \{q, q'\}$.
\end{proof}

\begin{ourlemma} \label{lem-extender}
For each $q \in Q$ one can compute in polynomial time a word $w_q \in \Sigma^*$ with $|w_q| \le \frac12 (c-1) (n+2) (n-1)$ such that no state $q' \ne q$ survives~$w_q$ and is coreachable with~$q$.
\end{ourlemma}
\begin{proof}
Let $q \in Q$.
Consider the following algorithm:

\begin{algorithmic}[1]
\STATE{$w := \varepsilon$}
\WHILE{there is $q' \in C(q)$ such that $q'$ survives~$w$}
\STATE{$w := w_{q,q'} w$ (with $w_{q,q'}$ from \cref{lem-product-automaton})}
\ENDWHILE
\RETURN $w_q := w$
\end{algorithmic}
The following picture visualizes aspects of this algorithm:
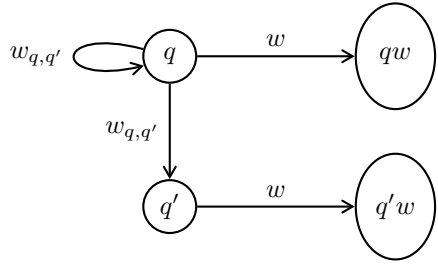
\begin{center}
\begin{tikzpicture}[xscale=3,yscale=2,AUT style]
\node[state] (q) at (0,0) {$q$};
\node[state] (qp) at (0,-1) {$q'$};
\node[state,ellipse,minimum height=40,minimum width=30] (qw) at (1,0) {$q w$};
\node[state,ellipse,minimum height=40,minimum width=30] (qpw) at (1,-1) {$q' w$};
\path[->] (q) edge [loop,out=160,in=200,looseness=9] node[left] {$w_{q,q'}$} (q);
\path[->] (q) edge node[left] {$w_{q,q'}$} (qp);
\path[->] (q) edge node[above] {$w$} (qw);
\path[->] (qp) edge node[above] {$w$} (qpw);
\end{tikzpicture}
\end{center}
We argue that the computed word~$w_q$ has the required properties.
First we show that the set $q w$ increases in each iteration of the algorithm.
Indeed, let $w$ and $w_{q,q'} w$ be the words computed by two subsequent iterations.
Since $q w_{q,q'} \supseteq \{q, q'\}$, we have $q w_{q,q'} w \supseteq q w \cup q' w$.
The set $q' w$ is nonempty, as $q'$ survives~$w$.
As can be read off from the picture above, the sets $q w$ and~$q' w$ are disjoint, as otherwise there would be two distinct paths from $q$ to a state in $q w \cap q' w$, both labelled by~$w_{q,q'} w$, contradicting unambiguousness.
It follows that $q w_{q,q'} w \supsetneq q w$.
Hence the algorithm must terminate.

Since in each iteration the set $q w$ increases by at least one element (starting from $\{q\}$), there are at most $c-1$ iterations.
Hence $|w_q| \le \frac12 (c-1)(n+2)(n-1)$.
%
There is no state $q' \ne q$ that survives~$w_q$ and is coreachable with~$q$, as otherwise the algorithm would not have terminated.
\end{proof}

\begin{ourlemma} \label{lem-extender-all}
One can compute in polynomial time words $z, y \in \Sigma^*$ such that:
\begin{itemize}
\item $|z| \le \frac14 (c-1)(n+2)n(n-1)$ and there are no two coreachable states that both survive~$z$;
\item $|y| \le \frac14 (m-1)(n+2)n(n-1)$ and there are no two mergeable   states that are both reached by~$y$.
\end{itemize}
\end{ourlemma}
\begin{proof}
As the two statements are dual, we prove only the first one.
Consider the following algorithm:

\begin{algorithmic}[1]
\STATE{$w := \varepsilon$}
\WHILE{there are coreachable $p, p'$ that both survive~$w$}
\STATE{$q := $ arbitrary state from $p w$}
\STATE{$w := w w_q$ (with $w_q$ from \cref{lem-extender})}
\ENDWHILE
\RETURN $z := w$
\end{algorithmic}
We show that the set
\[
 B \ := \ \{p_1 \in Q : \exists\, p_2 \in C(p_1) \text{ such that both } p_1,p_2 \text{ survive } w\}
\]
loses at least two states in each iteration.
First observe that
\[
 B' \ := \ \{p_1 \in Q : \exists\, p_2 \in C(p_1) \text{ such that both } p_1,p_2 \text{ survive } w w_q\}
\]
is clearly a subset of~$B$.

Let $p \in B$ be the state from line~$2$ of the algorithm, and let $q \in p w$ be the state from the body of the loop.
We claim that no $p'' \in C(p)$ survives $w w_q$.
Indeed, let $p'' \in C(p)$.
The following picture visualizes the situation:
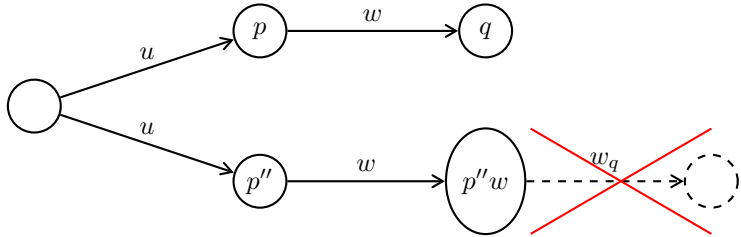
\begin{center}
\begin{tikzpicture}[xscale=3,yscale=1,AUT style]
\node[state] (o) at (0,0) {};
\node[state] (p) at (1,1) {$p$};
\node[state] (pp) at (1,-1) {$p''$};
\node[state] (q) at (2,1) {$q$};
\node[states] (ppw) at (2,-1) {$p'' w$};
\node[state,dashed] (f) at (3,-1) {};
\path[->] (o) edge node[above] {$u$} (p);
\path[->] (o) edge node[above] {$u$} (pp);
\path[->] (p) edge node[above] {$w$} (q);
\path[->] (pp)edge node[above] {$w$} (ppw);
\path[->,dashed] (ppw) edge node[above] {$w_q$} (f);
\draw[red,thick] (3.0,-0.3) edge (2.2,-1.7);
\draw[red,thick] (3.0,-1.7) edge (2.2,-0.3);
\end{tikzpicture}
\end{center}
By unambiguousness and since $q \in p w$, we have $q \not\in p'' w $.
By the definition of~$w_q$ and since all states in $p'' w$ are coreachable with~$q$, we have $p'' w w_q = \emptyset$, which proves the claim.

By the claim, we have $p \not\in B'$.
Let $p' \in B$ be the state~$p'$ from line~$2$ of the algorithm.
We have $p' \in C(p)$.
By the claim, $p'$ does not survive $w w_q$.
Hence $p' \not\in B'$.

So we have shown that the algorithm removes at least two states from~$B$ in every iteration.
Thus it terminates after at most $\frac n 2$ iterations.
Using the length bound from \cref{lem-extender} we get $|z| \le \frac14 (c-1)(n+2)n(n-1)$.
There are no coreachable $q,q'$ that both survive~$z$, as otherwise the algorithm would not have terminated.
\end{proof}

For the following development, let $q_1, \ldots, q_k$ be the states that are reached by~$y$ and survive~$z$ (with $y, z$ from \cref{lem-extender-all}), see \cref{fig-q1k}.
\begin{figure}[ht]
\begin{center}
\begin{tikzpicture}[xscale=3,yscale=2,AUT style]
\node[state] (q1) at (0,0) {$q_1$};
\node[state] (q2) at (0,-1) {$q_2$};
\node[state] (qk) at (0,-2.8) {$q_k$};
\node[states] (wpq1) at (-1,0) {$y q_1$};
\node[states] (wpq2) at (-1,-1) {$y q_2$};
\node[states] (wpqk) at (-1,-2.8) {$y q_k$};
\node[states] (q1z) at (1,0) {$q_1 z$};
\node[states] (q2z) at (1,-1) {$q_2 z$};
\node[states] (qkw) at (1,-2.8) {$q_k z$};
\foreach \i in {-1,0,1}{
\node[circle,fill,inner sep=1.0] at (\i,-1.7) {};
\node[circle,fill,inner sep=1.0] at (\i,-1.9) {};
\node[circle,fill,inner sep=1.0] at (\i,-2.1) {};
}
\path[->] (wpq1) edge node[above] {$y$} (q1);
\path[->] (wpq2) edge node[above] {$y$} (q2);
\path[->] (wpqk) edge node[above] {$y$} (qk);
\path[->] (q1) edge node[above] {$z$} (q1z);
\path[->] (q2) edge node[above] {$z$} (q2z);
\path[->] (qk) edge node[above] {$z$} (qkw);
\node[state,dashed] (l) at (-2,-0.5) {};
\node[state,dashed] (r) at (+2,-0.5) {};
\path[->,dashed] (l) edge node[above] {$u$} (wpq1);
\path[->,dashed] (l) edge node[above] {$u$} (wpq2);
\path[->,dashed] (q1z) edge node[above] {$u'$} (r);
\path[->,dashed] (q2z) edge node[above] {$u'$} (r);
\draw[red,thick] (-2.1,0.1) edge (-1.25,-1.1);
\draw[red,thick] (-2.1,-1.1) edge (-1.25,0.1);
\draw[red,thick] (2.1,0.1) edge (1.25,-1.1);
\draw[red,thick] (2.1,-1.1) edge (1.25,0.1);
\end{tikzpicture}
\end{center}
\caption{The states $q_1, \ldots, q_k$ are neither coreachable nor mergeable.}
\label{fig-q1k}
\end{figure}
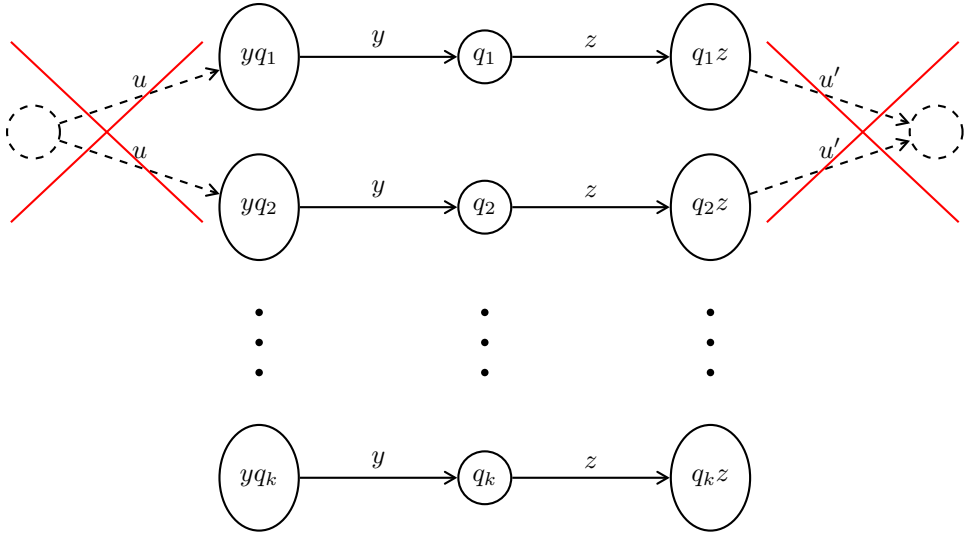
\begin{ourlemma} \label{lem-qi-basic}
Let $1 \le i < j \le k$.
Then $q_i, q_j$ are neither coreachable nor mergeable.
\end{ourlemma}
\begin{proof}
Immediate from the properties of $y, z$ (\cref{lem-extender-all}).
\end{proof}
The following lemma restricts sets of the form $q_i z x y z$ for $i \in \{1, \ldots, k\}$ and $x \in \Sigma^*$:
\begin{ourlemma} \label{lem-qi-no-distribution}
Let $i \in \{1, \ldots, k\}$ and $x \in \Sigma^*$.
Then there is $j \in \{1, \ldots, k\}$ such that $q_i z x y z \subseteq q_j z$.
\end{ourlemma}
\begin{proof}
If $q_i z x y z = \emptyset$ then choose $j$ arbitrarily.
Otherwise, let $q \in q_i z x y z$.
Then $q$ is reached by $y z$, so there is $j$ with $q_i z x y \ni q_j$ and $q_j z \ni q$.
We show that $q_i z x y z \subseteq q_j z$.
To this end, let $q' \in q_i z x y z$.
Then $q'$ is reached by~$y z$, so there is $j'$ with $q_i z x y \ni q_{j'}$ and $q_{j'} z \ni q'$.
Since $q_i z x y \supseteq \{q_j, q_{j'}\}$ and $q_j, q_{j'}$ are not coreachable 
(by \cref{lem-qi-basic}), we have $j' = j$.
Hence $q_j z = q_{j'} z \ni q'$.
\end{proof}
Provided that there is a killing word (which can be checked in polynomial time via \cref{prop-mortality}), the following lemma asserts that for each $i \in \{1, \ldots, k\}$ one can efficiently compute a short word~$x_i$ such that no state in $q_i z$ survives $x_i y z$.
The proof hinges on a linear-algebra technique for checking equivalence of automata that are weighted over a field.
The argument goes back to Sch\"utzenberger~\cite{Schutzenberger} and has often been rediscovered, see, e.g., \cite{Tzeng}.
\begin{ourlemma} \label{lem-UFA-strong-sense}
Suppose that $0 \in M(\Sigma^*)$. 
For each $i \in \{1, \ldots, k\}$ one can compute in polynomial time a word $x_i \in \Sigma^*$ with $|x_i| \le n$ such that $q_i z x_i y z = \emptyset$.
\end{ourlemma}
\begin{proof}
Let $i \in \{1, \ldots, k\}$.
Since $y \{q_1, \ldots, q_k\}$ are the only states to survive $y z$, it suffices to compute $x \in \Sigma^*$ with $|x| \le n$ such that $q_i z x \cap y \{q_1, \ldots, q_k\} = \emptyset$.

Define $e \in \{0,1\}^Q$ as the characteristic row vector of $q_i z$, i.e., $e(q) = 1$ if and only if $q \in q_i z$.
Define $f \in \{0,1\}^Q$ as the characteristic column vector $y \{q_1, \ldots, q_k\}$.
First we show that for any $x \in \Sigma^*$ we have $e M(x) f \le 1$. 
Towards a contradiction suppose $e M(x) f \ge 2$.
Then there are two distinct $x$-labelled paths from $q_i z$ to $y \{q_1, \ldots, q_k\}$.
It follows that there are two distinct $z x y$-labelled paths from $q_i$ to $\{q_1, \ldots, q_k\}$.
By unambiguousness, these paths end in two distinct states $q_j, q_{j'}$. 
But then $q_j, q_{j'}$ are coreachable, contradicting \cref{lem-qi-basic}.
Hence we have shown that $e M(x) f \le 1$ holds for all $x \in \Sigma^*$.

Define the (row) vector space
\[
V \ := \
\left\langle \begin{pmatrix} e M(x) & 1 \end{pmatrix} : x \in \Sigma^* \right\rangle \ \subseteq \ \R^{n+1}\;,
\]
i.e., $V$ is spanned by the vectors $\begin{pmatrix} e M(x) & 1 \end{pmatrix}$ for $x \in \Sigma^*$.
The vector space~$V$ can be equivalently characterized as the smallest vector space that contains $\begin{pmatrix} e & 1 \end{pmatrix}$ and is closed under multiplication with $\left(\begin{smallmatrix} M(a) & 0 \\ 0 & 1 \end{smallmatrix}\right)$ for all $a \in \Sigma$.
Hence the following algorithm computes a set $B \subseteq \Sigma^*$ such that
$
 \left\{\begin{pmatrix} e M(x) & 1 \end{pmatrix} : x \in B \right\}
$
is a basis of~$V$:

\begin{algorithmic}[1]
\STATE{$B := \{\varepsilon\}$}
\WHILE{$\exists\,u \in B,\,a \in \Sigma$ such that $\begin{pmatrix} e M(u a) & 1 \end{pmatrix} \not\in
  \left\langle \begin{pmatrix} e M(x) & 1 \end{pmatrix} : x \in B \right\rangle$}
\STATE{$B := B \cup \{u a\}$}
\ENDWHILE
\RETURN $B$ 
\end{algorithmic}
Observe that the algorithm performs at most $n$ iterations of the loop body, as every iteration increases the dimension of the space
$
 \left\langle \begin{pmatrix} e M(x) & 1 \end{pmatrix} : x \in B \right\rangle
$
by~$1$, but the dimension cannot grow larger than $n+1$.
Hence $|x| \le n$ holds for all $x \in B$.
Since $M(w_0) = 0$ holds for some $w_0 \in \Sigma^*$ and hence $e M(w_0) f = 0 \ne 1$, the space~$V$ is not orthogonal to $\left(\begin{smallmatrix} f \\ -1 \end{smallmatrix}\right)$.
So there exists $x \in B$ such that $e M(x) f \ne 1$.
Since $e M(x) f \le 1$, we have $e M(x) f = 0$.
Hence $q_i z x \cap y \{q_1, \ldots, q_k\} = \emptyset$.
\end{proof}

Now we can prove the following lemma, which is our main technical contribution:%
\begin{ourlemma} \label{lem-sc-main}
Suppose that $0 \in M(\Sigma^*)$.
One can compute in polynomial time a word $w \in \Sigma^*$ with $M(w) = 0$ and $|w| \le \ourbound$.
\end{ourlemma}
\begin{proof}
For any $1 \le j < j' \le k$ the sets $q_j z$ and $q_{j'} z$ are disjoint by \cref{lem-qi-basic} and nonempty.
Hence any $P' \subseteq Q$ has at most one set $P \subseteq \{q_1, \ldots, q_k\}$ with $P z = P'$, which we call the \emph{generator} of~$P'$. 
Note that all sets of the form $Q' y z$ where $Q' \subseteq Q$ have a generator.
For any $i \in \{1, \ldots, k\}$, let $x_i$ be the word from \cref{lem-UFA-strong-sense}, i.e., $q_i z x_i y z = \emptyset$.
By \cref{lem-qi-no-distribution}, for any $j \in \{1, \ldots, k\}$ the generator of $q_j z x_i y z$ has at most one element.
Thus, if $q_i \in P \subseteq \{q_1, \ldots, q_k\}$, then the generator, $P$, of $P z$ has strictly more elements than the generator of $P z x_i y z$.

Consider the following algorithm:

\begin{algorithmic}[1]
\STATE{$w := y z$}
\WHILE{$Q w \ne \emptyset$}
\STATE{$q_i := $ arbitrary element of the generator of $Q w$}
\STATE{$w := w x_i y z$}
\ENDWHILE
\RETURN $w$
\end{algorithmic}
It follows from the argument above that the size of the generator of $Q w$ decreases in every iteration of the loop.
Hence the algorithm terminates after at most $k$ iterations and computes a word~$w$ such that $Q w = \emptyset$ and, using \cref{lem-extender-all,lem-UFA-strong-sense},
\begin{align*} |w| &\ \le \ |y z| + k (n + |y z|) \ \le \ n^2 + (k+1)(|y|+|z|) \\ &\ \le \ n^2 + \frac14 (k+1) (c + m - 2) (n+2) n (n-1)\,.
\end{align*}

Let $q, q' \in Q$ and $u, u' \in \Sigma^*$ such that $c = |q u|$ and $m = |u' q'|$.
Clearly, $q u \cup u' q' \cup \{q_1, \ldots, q_k\} \subseteq Q$, and it follows from the inclusion-exclusion principle:
\[
 c + m + k \ \le \ n + |q u \cap u' q'| + |q u \cap \{q_1, \ldots, q_k\}| + |\{q_1, \ldots, q_k\} \cap u' q'|
\]
The sets $q u$ and $u' q'$ overlap in at most one state by unambiguousness.
The sets $q u$ and $\{q_1, \ldots, q_k\}$ overlap in at most one state by \cref{lem-qi-basic}, and similarly for $\{q_1, \ldots, q_k\}$ and $u' q'$.
It follows $c + m + k \le n + 3$, thus $(k+1) + (c+m-2) \le n+2$, hence $(k+1)(c+m-2) \le \frac14 (n+2)^2$.
With the bound on~$|w|$ from above we conclude that $|w| \le n^2 + \frac1{16} (n+2)^3 n (n-1)$, which is bounded by $\ourbound$ for $n \ge 1$.
\end{proof}

The following lemma, which rests on the properties of $y$ and~$z$, provides an alternative to the use of Carpi's \cref{thm-Carpi} in the proof of \cref{prop-sc}.
\begin{ourlemma} \label{lem-carpi-alt}
Suppose that $0 \not\in M(\Sigma^*)$.
Then $M(y z)$ has minimum rank in~$M(\Sigma^*)$ and this rank is~$k$.
\end{ourlemma}
\begin{proof}
It follows from \cref{lem-qi-basic} that each row of~$M(y z)$ is either the zero vector or the characteristic vector of some $q_i z$.
As the sets $q_i z$ for $i \in \{1, \ldots, k\}$ are nonempty and pairwise disjoint, it follows that $M(y z)$ has rank~$k$.

Suppose $x \in \Sigma^*$ is such that $M(x)$ has rank less than~$k$.
Then $M(y z x)$ has rank less than~$k$.
Since the sets $q_i z x$ for $i \in \{1, \ldots, k\}$ are pairwise disjoint, there is $i \in \{1, \ldots, k\}$ such that $q_i z x = \emptyset$.
In order to show that $0 \in M(\Sigma^*)$ it suffices to show that for all $p \in Q$ and all $u \in \Sigma^*$ there is $w \in \Sigma^*$ such that $p u w = \emptyset$.
Let $p \in Q$ and $u \in \Sigma^*$.
If $p u = \emptyset$ then choose $w = \varepsilon$.
Otherwise, let $v \in \Sigma^*$ be such that $p u v y \ni q_i$.
By \cref{lem-qi-basic}, we have $p u v y \cap \{q_1, \ldots, q_k\} = \{q_i\}$.
Thus $p u v y z = q_i z$ and $p u v y z x = q_i z x = \emptyset$.
Hence choose $w = v y z x$.
\end{proof}

To prove \cref{prop-sc} we combine \cref{lem-sc-main} with either Carpi's \cref{thm-Carpi} or \cref{lem-carpi-alt}.
\propsc*
\begin{proof}
One can check in polynomial time whether 
$0 \in M(\Sigma^*)$,
see \cref{prop-mortality}.
If yes, then the minimum rank is~$0$, and \cref{lem-sc-main} gives the result.
Otherwise, $0 \not\in M(\Sigma^*)$, and \cref{lem-extender-all,lem-carpi-alt} give the result.

In the case $0 \not\in M(\Sigma^*)$ one may alternatively use Carpi's \cref{thm-Carpi}:
Indeed, the minimum rank~$r$ is between $1$ and~$n$, and hence $n \ge 1$.
\cref{thm-Carpi} asserts the existence of a word~$w$ such that $M(w)$ has rank~$r$ and $|w| \le \frac12 n^4 - n^3 + \frac52 n^2 - 3 n + 1$, which is bounded by $\ourbound$ for $n \ge 1$.
An inspection of Carpi's proof~\cite{Carpi88} shows that his proof is constructive and can be transformed into an algorithm that computes~$w$ in polynomial time.
\end{proof}

\subsection{Not Necessarily Strongly Connected} \label{sub-nsc}

We prove \cref{thm-main}:
\begin{ourtheorem}[continues=thm-main]
\stmtthmmainrephrased
\end{ourtheorem}
\begin{proof}
For any matrix~$A$ denote by $\rk(A)$ its rank.
For $i,j \in \{1, \ldots, n\}$ write $i \to j$ if there is $u \in \Sigma^*$ such that $M(u)(i,j) > 0$, and write $i \leftrightarrow j$ if $i \to j$ and $j \to i$.
The relation ${\leftrightarrow}$ is an equivalence relation.
Denote by $C_1, \ldots, C_h \subseteq \{1, \ldots, n\}$ its equivalence classes ($h \le n$).
We can assume that whenever $i \in C_k$ and $j \in C_\ell$ and $i \to j$, then $k \le \ell$.
Hence, without loss of generality, $M(u)$ for any $u \in \Sigma^*$ has the following block-upper triangular form:
\[
M(u) \ = \ \begin{pmatrix}
M_{11}(u) & M_{12}(u) & \cdots & M_{1 h}(u) \\
  0       & M_{22}(u) & \cdots & M_{2 h}(u) \\
  \vdots  & \vdots    & \ddots & \vdots     \\
  0       & 0         & \cdots  & M_{h h}(u)
\end{pmatrix}\;,
\]
where $M_{i i}(u) \in \N^{|C_i| \times |C_i|}$ for all $i \in \{1, \ldots, h\}$.
For $i \in \{1, \ldots, h\}$ define $r_i := \min_{u \in \Sigma^*} \rk(M_{i i}(u))$.
For any $u \in \Sigma^*$ we have $\rk(M(u)) \ge \sum_{i=1}^h \rk(M_{i i}(u))$ (see, e.g., \cite[Chapter 0.9.4]{HornJohnson13}).
It follows that the minimum rank among the matrices in $M(\Sigma^*)$ is at least $\sum_{i=1}^h r_i$.

Let $w_1, \ldots, w_h \in \Sigma^*$ be the words from \cref{prop-sc} for $M_{1 1}, \ldots, M_{h h}$, respectively, so that $\rk(M_{i i}(w_i)) = r_i$ holds for all $i \in \{1, \ldots, h\}$.
Define $w := w_1 \cdots w_h$.
Then we have:
\[
|w| \ \le \ \sum_{i=1}^h |w_i| \ \le \ \sum_{i=1}^h \frac1{16} |C_i|^5 + \frac{15}{16} |C_i|^4 \ \le \ \ourbound
\]

It remains to show that $\rk(M(w)) \le \sum_{i=1}^h r_i$.
It suffices to prove that $\rk(M_k(w_1 \cdots w_k)) \le \sum_{i=1}^k r_i$ holds for all $k \in \{1, \ldots, h\}$, where $M_k(u)$ for any $u \in \Sigma^*$ is the principal submatrix obtained by restricting $M(u)$ to the rows and columns corresponding to $\bigcup_{i=1}^k C_i$.
We proceed by induction on~$k$.
For the base case, $k=1$, we have $\rk(M_1(w_1)) = \rk(M_{1 1}(w_1)) = r_1$.
For the induction step, let $1 < k \le h$.
Then there are matrices $A_1, A_2, B_1, B_2$ such that:
\begin{align}
  M_k(w_1 \cdots w_k)
  &\ = \ M_k(w_1 \cdots w_{k-1}) M_k(w_k) \nonumber\\
  &\ = \ \begin{pmatrix} M_{k-1}(w_1 \cdots w_{k-1}) & A_1 \\ 0 & A_2 \end{pmatrix}
         \begin{pmatrix} B_1 & B_2 \\ 0 & M_{k k}(w_k) \end{pmatrix} \nonumber\\
  &\ = \ \begin{pmatrix} M_{k-1}(w_1 \cdots w_{k-1}) \\ 0 \end{pmatrix}
         \begin{pmatrix} B_1 & B_2 \end{pmatrix} +
         \begin{pmatrix} A_1 \\ A_2 \end{pmatrix}
         \begin{pmatrix} 0 & M_{k k}(w_k) \end{pmatrix} \label{eq-product-sum}
\end{align}
By the induction hypothesis, we have $\rk(M_{k-1}(w_1 \cdots w_{k-1})) \le \sum_{i=1}^{k-1} r_i$.
Further, we have $\rk(M_{k k}(w_k)) = r_k$.
So the ranks of the two summands in~\eqref{eq-product-sum} are at most $\sum_{i=1}^{k-1} r_i$ and $r_k$, respectively.
Since for any matrices $A, B$ it holds $\rk(A+B) \le \rk(A) + \rk(B)$, we conclude that $\rk(M_k(w_1 \cdots w_k)) \le \sum_{i=1}^k r_i$, completing the induction proof.
\end{proof}

\section{Proof of \texorpdfstring{\cref{thm-restivo}}{Theorem~\ref{thm-restivo}}} \label{sec-proof-restivo}

\thmrestivo*

Consider the flower automaton associated to~$X$, which is a UFA $\A = (\Sigma, Q, \delta)$ with $n = |Q| = m - |X| + 1$ states.
The uncompletable words in~$X$ are exactly the killing words in~$\A$.
Towards a proof of \cref{thm-restivo} we first focus on computing a short killing word in~$\A$.

To this end we optimize the construction from \cref{sub-sc} for flower automata.
Denote by $0 \in Q$ the ``central'' state of~$\A$ around which the petals are built.
For each $q \in Q$ fix a word $u_q \in \Sigma^*$ such that $q u_q = \{0\}$ and $|u_q| \le k-1$.
The following lemma bounds the size of certain sets of states that survive long words:
\begin{ourlemma} \label{lem-restivo-basic-lemma}
Let $w \in \Sigma^*$ with $|w| \ge k-1$.
Then for all $p \in Q$ and all $v \in \Sigma^*$ we have $|p v \cap w Q| \le k$, i.e., at most $k$~states of $p v$ survive~$w$.
\end{ourlemma}
\begin{proof}
Towards a contradiction, suppose $|p v \cap w Q| > k$.
By the pigeonhole principle, there are two different states $q_1, q_2$ with $q_1, q_2 \in p v \cap w Q$ such that $|u_{q_1}| = |u_{q_2}|$.
Since $q_1, q_2$ both survive $w$, where $|w| \ge k-1$, both $u_{q_1}$ and $u_{q_2}$ are prefixes of~$w$.
Hence $u_{q_1} = u_{q_2}$.
It follows that $q_1 u_{q_1} = \{0\} = q_2 u_{q_1}$, i.e., $q_1, q_2$ are mergeable.
But $q_1, q_2 \in p v$ are also coreachable, contradicting unambiguousness.
\end{proof}

The following lemma adapts \cref{lem-extender}:

\begin{ourlemma} \label{lem-extender-restivo}
For each $q \in Q$ one can compute in polynomial time a word $w_q \in \Sigma^*$ with $|w_q| \le \frac12 (k-1) (n+2) n$ such that no state $q' \ne q$ survives~$w_q$ and is coreachable with~$q$.
\end{ourlemma}
\begin{proof}
We use the same algorithm as in the proof of \cref{lem-extender}, except that we initialize $w$ not to $\varepsilon$ but to some $\winit \in \Sigma^{k-1}$ with $q \winit \ne \emptyset$.
Such $\winit$ exists, as $q$ is on some cycle.
Let $\ell \in \N$ be the number of iterations of the loop in the algorithm.
The computed word $w_q$ has the form $w_{q,q_\ell} \cdots w_{q,q_1} \winit$ for some states $q_1, \ldots, q_\ell \in Q$.
It follows from the proof of \cref{lem-extender} that for all $i \in \{1, \ldots, \ell\}$ we have $q w_{q,q_i} \cdots w_{q,q_1} \winit \supsetneq q w_{q,q_{i-1}} \cdots w_{q,q_1} \winit$.
Hence also $q w_{q,q_i} \cdots w_{q,q_1} \cap \winit Q \supsetneq q w_{q,q_{i-1}} \cdots w_{q,q_1} \cap \winit Q$.
Since $q \in \winit Q$, it follows that $|q w_{q,q_\ell} \cdots w_{q,q_1} \cap \winit Q| \ge \ell + 1$.
By \cref{lem-restivo-basic-lemma} it follows $\ell + 1 \le k$.
Hence, using \cref{lem-product-automaton} we obtain $|w_q| = |w_{q,q_\ell} \cdots w_{q,q_1} \winit| \le \frac12 (k-1) (n+2) (n-1) + (k-1) \le \frac12 (k-1) (n+2) n$.
\end{proof}

The following lemma adapts \cref{lem-extender-all}:

\begin{ourlemma} \label{lem-extender-all-restivo}
One can compute in polynomial time words $z, y \in \Sigma^*$ such that:
\begin{itemize}
\item $|z| \le \frac12 k (k-1) (n+2) (n+1)$ and there are no two coreachable states that both survive~$z$;
\item $|y| \le \frac12 k (k-1) (n+2) (n+1)$ and there are no two mergeable   states that are both reached by~$y$.
\end{itemize}
\end{ourlemma}
\begin{proof}
We use the same algorithm as in the proof of \cref{lem-extender-all}, except that we initialize $w$ not to $\varepsilon$ but to an arbitrary $\winit \in \Sigma^{k-1}$, and that for~$w_q$ we use \cref{lem-extender-restivo} instead of \cref{lem-extender}:
\begin{algorithmic}[1]
\STATE{$w := \winit$}
\WHILE{there are coreachable $p, p'$ that both survive~$w$}
\STATE{$q := $ arbitrary state from $p w$}
\STATE{$w := w w_q$ (with $w_q$ from \cref{lem-extender-restivo})}
\ENDWHILE
\RETURN $z := w$
\end{algorithmic}
Consider a state $p \in Q$ picked in some iteration of the loop, i.e., $p$ survives the (current) word~$w$.
We claim that no state $\bar{p}$ with $|u_{\bar{p}}| = |u_p|$ will be picked in any future iteration.
Indeed, let $\bar{p} \in Q$ be with $|u_{\bar{p}}| = |u_p|$ such that $\bar{p}$ survives a future~$w$.
Then $\bar{p}$ survives the current~$w$.
Since $|u_p| = |u_{\bar{p}}|$ and $p, \bar{p}$ both survive~$w$ with $|w| \ge k-1$, we have $u_p = u_{\bar{p}}$ and this word is a prefix of~$w$.
It follows that $p w = \bar{p} w$, thus $q \in \bar{p} w$, where $q$ is the state from line~3.
Suppose $\bar{p}'$ is an arbitrary state that is coreachable with~$\bar{p}$.
Then the states in $\bar{p}' w$ are coreachable with~$q$.
Thus, $\bar{p}' w w_q = \emptyset$ and so $\bar{p}'$ does not survive any future~$w$.
It follows that $\bar{p}$ will not be picked in any future iteration.

Since for all $p \in Q$ we have $|u_{p}| \in \{0, \ldots, k-1\}$, the algorithm performs at most $k$ loop iterations.
Hence, using \cref{lem-extender-restivo}, the computed word~$z$ has length at most $k \cdot \frac12 (k-1) (n+2) n + (k-1) \le \frac12 k (k-1) (n+2) (n+1)$.
The argument for~$y$ is similar.
\end{proof}

The following lemma adapts \cref{lem-sc-main}:

\begin{ourlemma} \label{lem-sc-main-restivo}
One can compute in polynomial time a killing word of length at most $(k+1) k^2 (n+2) (n+1)$.
\end{ourlemma}
\begin{proof}
Let $z,y$ be the words from \cref{lem-extender-all-restivo}.
We can assume that $|z| \ge k-1$.
First we argue that there are at most $k$ states that are reached by~$y$ and survive~$z$.
Towards a contradiction, suppose otherwise.
By the pigeonhole principle, there are two distinct states $q, q' \in Q$ that are reached by~$y$ and survive~$z$ and satisfy $|u_q| = |u_{q'}|$.
Since $|z| \ge k-1$, it follows that $u_q = u_{q'}$ is a prefix of~$z$, thus $q,q'$ are mergeable.
But $q,q'$ are reached by~$y$, contradicting the definition of~$y$.

It follows that the $k$ from \cref{sub-sc} is at most the $k$ from this section.
Mirroring exactly the proof of \cref{lem-sc-main} and using \cref{lem-extender-all-restivo}, we obtain a killing word~$w$ of length at most
\begin{align*} |w| &\ \le \ |y z| + k (n + |y z|) \ \le \ n^2 + (k+1)(|y|+|z|) \\ &\ \le \ n^2 + (k+1) k (k-1) (n+2) (n+1) \ \le \ (k+1) k^2 (n+2) (n+1)\,.
\end{align*}
\end{proof}

Finally we prove \cref{thm-restivo}:
\begin{proof}[Proof of \cref{thm-restivo}]
Since $k > 0$, it follows $|X| \ge 1$ and thus $n = m - |X| + 1 \le m$.
The result follows from \cref{lem-sc-main-restivo}.
\end{proof}

\section{Proof of \texorpdfstring{\cref{thm-long-product}}{Theorem~\ref{thm-long-product}}} \label{sec-proof-long-product}

\thmlongproduct*
\begin{proof}
Denote by $p_i$ the $i$th prime number (so $p_1 = 2$).
Let $m \ge 1$.
Define:
\begin{align*}
\Sigma \ &:= \ \{a, b_1, \ldots, b_m\} \\
Q_i \    &:= \ \{(i,0), (i,1), \ldots, (i,p_i-1)\} \quad \text{for every } i \in \{1, \ldots, m\} \\
Q \      &:= \ \{0\} \cup \bigcup_{i=1}^m Q_i
\end{align*}
Further, define a monoid morphism $M : \Sigma^* \to \N^{Q \times Q}$ by setting for all $i \in \{1, \ldots, m\}$
\begin{align*}
M(a)(0,(i,0)) \ &:= \ 1 \\
M(a)((i,j),(i,j+1 \bmod p_i)) \ &:= \ 1 \quad\text{ for all } j \in \{0, \ldots, p_i-1\} \\
M(b_i)(0,0) \ &:= \ 1 \\
M(b_i)((i,j),0) \ &:= \ 1 \quad\text{ for all } j \in \{0, \ldots, p_i-1\}
\end{align*}
and setting all other entries of $M(a), M(b_1), \ldots, M(b_m)$ to~$0$, see \cref{fig-M3}.
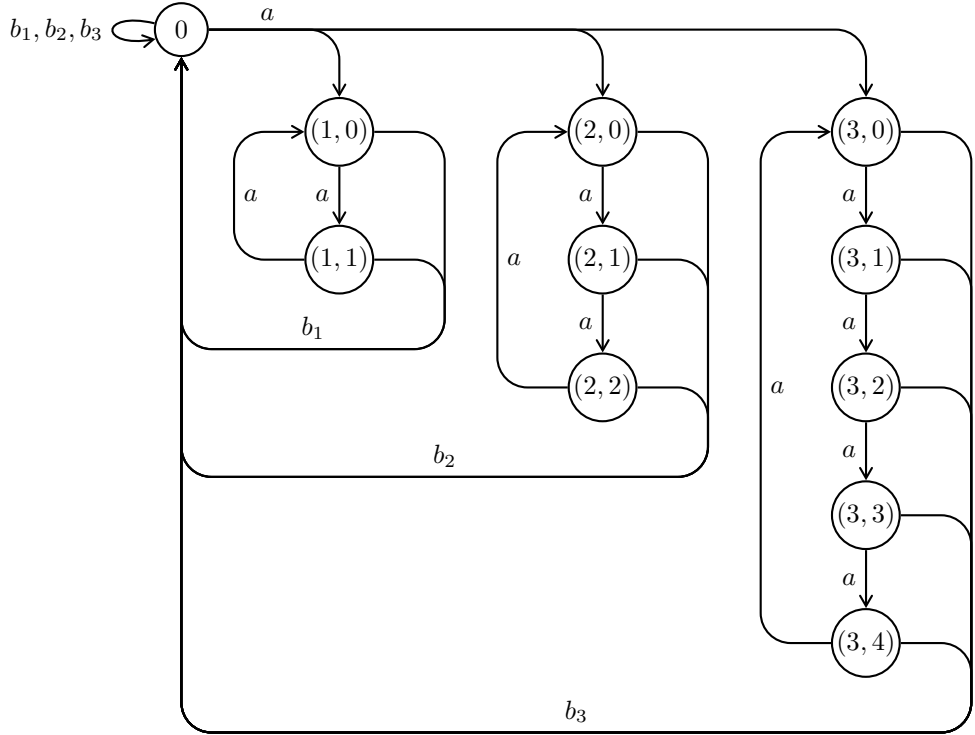
\begin{figure}
\begin{center}
\begin{tikzpicture}[xscale=3.5,yscale=1.7,AUT style]
\node[state] (0)  at (0.4,0.8) {$0$};
\node[state] (10) at (1,0) {$(1,0)$};
\node[state] (11) at (1,-1) {$(1,1)$};
\node[state] (20) at (2,0) {$(2,0)$};
\node[state] (21) at (2,-1) {$(2,1)$};
\node[state] (22) at (2,-2) {$(2,2)$};
\node[state] (30) at (3,0) {$(3,0)$};
\node[state] (31) at (3,-1) {$(3,1)$};
\node[state] (32) at (3,-2) {$(3,2)$};
\node[state] (33) at (3,-3) {$(3,3)$};
\node[state] (34) at (3,-4) {$(3,4)$};
\coordinate (r1) at (1.4,-1.7);
\coordinate (l1) at (0.4,  -1.7);
\coordinate (r2) at (2.4,-2.7);
\coordinate (l2) at (0.4,  -2.7);
\coordinate (r3) at (3.4,-4.7);
\coordinate (l3) at (0.4,  -4.7);
\path[->] (0) edge [loop,out=160,in=200,looseness=4] node[left] {$b_1,b_2,b_3$} (0);
\path[->] (10) edge node[left] {$a$} (11);
\path[->] (20) edge node[left] {$a$} (21);
\path[->] (21) edge node[left] {$a$} (22);
\path[->] (30) edge node[left] {$a$} (31);
\path[->] (31) edge node[left] {$a$} (32);
\path[->] (32) edge node[left] {$a$} (33);
\path[->] (33) edge node[left] {$a$} (34);
\draw[->,rounded corners=4mm] (11) -- +(-0.4,0) -- node[right] {$a$} ($ (10) + (-0.4,0) $) -- (10);
\draw[->,rounded corners=4mm] (22) -- +(-0.4,0) -- node[right] {$a$} ($ (20) + (-0.4,0) $) -- (20);
\draw[->,rounded corners=4mm] (34) -- +(-0.4,0) -- node[right] {$a$} ($ (30) + (-0.4,0) $) -- (30);
\draw[->,rounded corners=4mm] (11) -| (r1) -- node[above] {$b_1$} (l1) -- (0);
\draw[->,rounded corners=4mm] (10) -| (r1) -| (0);
\draw[->,rounded corners=4mm] (22) -| (r2) -- node[above] {$b_2$} (l2) -- (0);
\draw[->,rounded corners=4mm] (21) -| (r2) -| (0);
\draw[->,rounded corners=4mm] (20) -| (r2) -| (0);
\draw[->,rounded corners=4mm] (34) -| (r3) -- node[above] {$b_3$} (l3) -- (0);
\draw[->,rounded corners=4mm] (33) -| (r3) -| (0);
\draw[->,rounded corners=4mm] (32) -| (r3) -| (0);
\draw[->,rounded corners=4mm] (31) -| (r3) -| (0);
\draw[->,rounded corners=4mm] (30) -| (r3) -| (0);
\draw[->,rounded corners=4mm] (0) -- node[above,pos=0.45] {$a$} +(0.6,0) -- (10);
\draw[->,rounded corners=4mm] (0) -| (20);
\draw[->,rounded corners=4mm] (0) -| (30);
\end{tikzpicture}
\end{center}
\caption{Automaton representation of $M$ for $m=3$.}
\label{fig-M3}
\end{figure}
We have $M(\Sigma^*) \subseteq \{0, 1\}^{Q \times Q}$, i.e., $M(\Sigma^*)$ is an unambiguous monoid of relations.
For all $q \in Q$ and all $q' \in Q \setminus \{0\}$ we have $M(b_1)(q,q') = 0$, i.e., $M(b_1)$ has rank~$1$.
For all $w \in \Sigma^*$ there is $q \in Q$ with $M(w)(0,q) = 1$, i.e., $1$ is the minimum rank in $M(\Sigma^*)$.
A shortest word $w_0 \in \Sigma^*$ such that $M(w_0)$ has rank~$1$ and $M(w_0)(0,(i,p_i-1)) = 1$ holds for all $i \in \{1, \ldots, m\}$ is the word~$w_0 = b_1 a^P$ where $P = \prod_{i=1}^m p_i \ge 2^m$.
On the other hand, we have $|Q| = 1 + \sum_{i=1}^m p_i \in O(m^2 \log m)$ by the prime number theorem.

Hence there is no polynomial~$p$ such that $P \le p(|Q|)$ holds for all~$m$.
\end{proof}

\bibliographystyle{siamplain}
\bibliography{killingword}

\end{document}